\documentclass[11pt]{article}
\usepackage{graphicx,amssymb,amsmath,amsthm}
\usepackage[a4paper,margin=25mm]{geometry}
\usepackage{mathtools}
\usepackage{wasysym}
\usepackage{stmaryrd}
\usepackage{hyperref}
\usepackage{xcolor}
\usepackage[american]{babel}
\usepackage{enumitem}
\usepackage[disable]{todonotes}
\newcommand{\therese}[1]{}
\newcommand{\guenter}[1]{}
\newcommand{\erin}[1]{}
\newcommand{\changed}{}
\graphicspath{{figures/}}

\newtheorem{theorem}{Theorem}
\newtheorem{lemma}{Lemma}
\newtheorem{corollary}{Corollary}
\newtheorem{definition}{Definition}
\newtheorem{claim}{Claim}
\newtheorem{observation}{Observation}
\newtheorem{proposition}{Proposition}
\newcommand{\calA}{{\ensuremath{\cal A}}}

\begin{document}

\iftrue  % switch iftrue/iffalse to see/not see identifying information
\title{Sweeping $x$-monotone pseudolines\thanks{Research initiated during the 2020 Workshop on Graphs and Geometry at the Bellairs Research Institute.  The authors would like to thank all participants of the workshop, but especially Hugo Akitaya and Stefan Felsner, for helpful input.}}
\author{Therese Biedl
\thanks{David R.~Cheriton School of Computer
Science, University of Waterloo,
%
%Waterloo, % repeated, could save one line
%
%Ontario N2L 3G1, 
Canada.
Supported by NSERC. }
\and Erin Chambers
\thanks{Department of Computer Science and Engineering, University of Notre Dame, Indiana, USA.
Research supported in part by NSF awards 1907612 and 2444309. }
\and Irina Kostitsyna
\thanks{KBR at NASA Ames Research Center, Moffett Field, CA, USA.}
\and G\"unter Rote\thanks{Freie Universit\"at Berlin, Institut f\"ur Informatik, Germany. }
}
\else
\fi
%\date{\today}

\pagestyle{myheadings}
\markboth{}{Th.\ Biedl, E. Chambers, I. Kostitsyna and G.\ Rote: \
  Sweeping $x$-monotone pseudolines}

\maketitle

\begin{abstract}
We study the problem of sweeping a pseudoline arrangement with $n$
$x$-monotone curves with a rope (an $x$-monotone curve that connects the points
at infinity).  The rope can move by flipping over
a face of the arrangement, replacing parts of it
from the lower to the upper chain of the face. % defined by the arrangement.
Counting as length of the rope the number of edges, what rope-length can be
needed in such a sweep?   We show that all such arrangements can be swept
with rope-length at most $2n-2$, and for some arrangements rope-length 
at least $\tfrac{7}{4}(n-2)+1$ is required.  We also discuss some complexity
issues around the problem of computing a sweep with the shortest rope-length.
  
\end{abstract}

%\tableofcontents
\section{Introduction}

Consider an arrangement $\calA$ of $n$ $x$-monotone infinite curves
where each pair of curves crosses % each other
exactly once. 
These define a directed acyclic planar
graph $G_\calA$, by replacing each crossing with a new vertex, adding
two vertices $s,t$ at negative and positive infinity, and directing edges left-to-right.  
This paper concerns
the problem of sweeping the arrangement with a rope of short length, or
equivalently, sweeping $G_\calA$ with a sequence of short $st$-paths.
Formally, we start with a rope at the lower hull of the arrangement.
At each step, whenever the rope contains the bottom chain
of an inner face $F$, we may \emph{flip} across $F$ by replacing the bottom chain
by the top chain of $F$. We stop when the rope is the upper hull.
The \emph{rope-length} of such a sweep is the maximum length of the rope, 
measured as the number of edges in the graph.  See Figure~\ref{fig:nFace}.

\begin{figure}[ht]
\centering%
\includegraphics[page=1,scale=1.2,%width=\linewidth,
trim=0 30 0 30,clip]{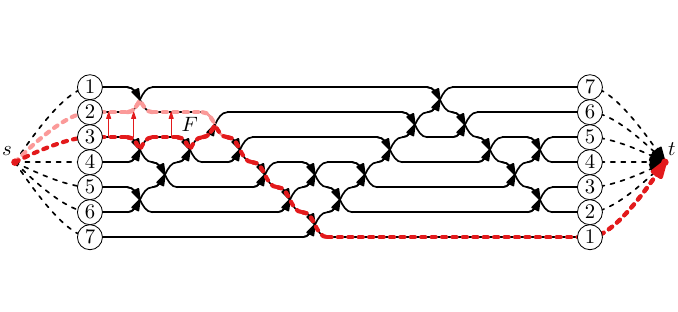}%
\label{fig:flip-face}%
\caption{
A pseudoline arrangement $\calA$ with seven $x$-monotone curves 
and the corresponding graph $G_\calA$.
Rope $\pi$ (red dashed) has length~8 and can be flipped across face $F$.
} 
\label{fig:nFace}
\end{figure}

One can easily construct an arrangement $\calA$ where the lower hull has length
$n$, so we cannot in general hope to find a sweep of rope-length less than $n$.   
But can we always achieve rope-length $n+O(1)$ with a suitable sweep?   
%\therese{Guenter, in the old writeup you claimed that this is false if three pseuodlines may meet in a point, since then supposedly the lower hull can have length $\Omega(n\alpha(n))$.   This can't be right; the indices are always  decreasing along the lower hull and so it has length at most $n$.   Do you remember what you were thinking of?   Maybe level sets? G: The only thing that I can think of
%that matches this remark is that  $\Theta(n\alpha(n))$ is the bound for the lower envelope of n SEGMENTS. }
We show that this is false:
for some arrangements we need rope-length at least $\tfrac{7}{4}n-\tfrac{5}{4}$.
We also provide an asymptotically matching upper bound:   For any such
arrangement $\calA$, we can find a sweep with rope-length at most $2n-2$.
Furthermore, the sweep has special properties: we 
simultaneously  sweep the dual graph $G_\calA^*$ of $G_\calA$,
and the two ropes of the two sweeps ``hug'' in some sense.

Finally, we study hardness results.   
A rope in $G_\calA$
corresponds to an edge-cut in $G_\calA^*$, and sweeping with a rope hence
corresponds to finding a vertex order that has small cuts.   This is the
\emph{cutwidth} problem, and since we impose special conditions on the graph
and the sweep, our problem is equivalent to solving {\sc Directed Cutwidth} in
$G_\calA^*$ (definitions and details are in Section~\ref{sec:hardness}).
Surprisingly enough, we have not been able to find NP-hardness
results for this problem, especially not in planar graphs.   We therefore show
that {\sc Directed Cutwidth} is NP-hard even in planar graphs with maximum
degree~6.   
Unfortunately the graphs constructed in the reduction are not duals of pseudoline
arrangements, so the complexity of minimizing the rope-length in
our sweeping problem remains open.

\paragraph{Related results:}   
\changed{The problem of minimizing the rope-length of a sweep is motivated by the problem of enumerating all arrangements
of $n$ pseudolines~\cite{Rote25}}.
An easy upper bound on the rope-length
in a sweep is
the maximum length of an $x$-monotone $st$-path.   However, this does
not lead to a good upper bound:
$x$-monotone paths can have close to $n^2$ edges
 \cite{Mat91,BRSSS04}, see
 \cite{Dum05} for related results.
This shows that it is necessary to choose a sweep carefully.
% However, this does not give a lower bound for
%our problem because our sweep need not
%visit these longest paths.

The idea of ``sweeping a plane graph'' is closely related to the so-called 
\emph{homotopy height}, see \cite{CMO18,BCEMO19,Ophelders2022} for an overview.
Here we are given an undirected planar graph $G$ with a fixed planar embedding and two
vertices $s,t$ on the outer-face.    We are asked to find a sequence of $st$-paths
that begin and end with the two $st$-paths that run along the outer-face.   Consecutive
$st$-paths in the sequence must be related via a limited set of \emph{permitted operations},
which include flipping across a face and introducing or eliminating a spike
along an edge.    The goal is to minimize the maximum path-length in the sequence.
Our problem is hence the same as computing the homotopy height, except
that we % severely
restrict the set of permitted operations since the path must follow
the edge directions. % always be directed, but in exchange also have a much more restrictedgraph.

Computing the homotopy height of a graph is in NP \cite{CMO18}, but it is
open whether this problem is NP-hard.     

With both face and spike moves, it is possible to prove that each path in the sequence can be assumed to be weakly simple, and under some restrictions on the input, the sequence of paths is \emph{monotone} in the sense that every face is swept exactly once \cite{CCMOR21}.  
In our setting, where only face moves are allowed, it is unclear if the optimal homotopy will be necessarily monotone, although this seems quite likely to be true; %that said, 
the proof in \cite{CCMOR21} relies upon spike moves as well as an underlying Riemannian metric structure for the disk, so it does not readily %seem to
apply in our setting.  

Another concept that
 is somewhat related to our sweep are ``strictly northward b-migrations,'' which were studied in the context of chains on lattices by Brightwell and Winkler~\cite{Brightwell2009}; while they are able to prove that their version is not monotone, their setting is not equivalent to arrangements of pseudolines, and hence does not answer the monotonicity question for our setting.

There is also a relationship between the homotopy height and the height of a planar
straight-line grid-drawing \cite{BCEMO19}; in particular 
this implies that  % TB: removed to save a line
every $N$-vertex planar graph $G$ has homotopy height at most $\tfrac{2}{3}N+O(1)$ since $G$ has a planar
straight-line grid-drawing where the smaller dimension is $\tfrac{2}{3}N+O(1)$ \cite{CN98}.
Unfortunately, this does not help to solve our problem, for two reasons.   First, in our sweeps 
we impose stronger restrictions on when we are allowed to flip across a face.
Second, we are sweeping an arrangement of $n$ curves, hence
the corresponding planar graph has $N\in\Theta(n^2)$ vertices and the above bounds are 
meaninglessly big.

As mentioned earlier, sweeping a pseudoline arrangement $\calA$ with a short rope
corresponds to solving {\sc Directed Cutwidth} in the dual graph $G_\calA^*$. 
The (undirected) version {\sc Cutwidth} of this problem is very well-established
in the literature and is known to be NP-hard even in planar graphs with maximum
degree~3 \cite{MonienS88}.
{\sc Cutwidth} is also SSE-hard to approximate within any constant factor
\cite{WuAPL14}.
SSE stands for the \emph{Small Set Expansion conjecture};
we refer the reader to this paper for the definition of ``SSE-hard''
and
other results concerning cutwidth.

%\therese{Somewhat loosely related is a paper on realizing tangles (https://arxiv.org/abs/2312.16213), i.e., you are given which pairs of curves have to cross (with maybe some restrictions on the order?  I did not read deeply...) and (translated to our world) you want a wire diagram realizing  this that has minimum width.   This seems rather different from what we do (we have no choice of the pseudoline arrangements), and so I don't think we need to cite this, but speak up if you know that paper better and see regions of overlap.}

\section{Definitions}

Throughout the paper, $\calA$ denotes a set of $n$ $x$-monotone
infinite curves that form a \emph{pseudoline arrangement}, i.e., each
pair of curves has exactly one point in common where the curves
properly cross.
The curves in
$\calA$ are called  \emph{pseudolines}.
Arrangement $\calA$ naturally defines a planar graph $G_\calA$, by replacing every crossing between pseudolines by a vertex, adding an edge whenever two crossings are consecutive on a pseudo-line, adding two vertices $s$ and $t$ that represent the points at negative and positive infinity, and connecting $s$ to the first crossing and $t$ to the last crossing of each pseudo-line.   We direct all edges of $G_\calA$ from left to right, making it a directed acyclic planar graph with exactly one source $s$ and one sink $t$ that are both on the outer-face.    Such a graph is known as a \emph{bipolar orientation},
and many properties are known, see for example \cite{FOR95}.    In particular, 
for any inner face $F$, the boundary consists of two directed paths; in our situation where edges are drawn left-to-right these paths naturally are called the \emph{top chain} and \emph{bottom chain} of~$F$. Their common start-vertex is the \emph{source} $s(F)$ of~$F$, and their common end-vertex is the \emph{sink} $t(F)$ of~$F$.
At any vertex $v\neq s,t$, the incoming edges are consecutive in the clockwise order around $v$, as are the outgoing edges.     In our situation with edges  drawn left-to-right, we can naturally speak of the topmost/bottommost incoming/outgoing edge of a vertex.

A \emph{rope} of $\calA$ is a directed $st$-path $\pi$ in $G_\calA$;   alternatively we can view $\pi$ as an $x$-monotone infinite curve along pseudo-lines.    For any two points $p,p'$ on $\pi$, we use $\pi(p,p')$ to denote the sub-curve between the two points (including $p,p'$). If $\pi$ contains the entire bottom chain of
some inner face $F$, then \emph{flipping rope $\pi$ across $F$} means to create a new rope that is $\pi$ except that the bottom chain $\pi(s(F),t(F))$ of $F$ gets replaced by the top chain of~$F$.    A \emph{sweep} of $\calA$ consists of a sequence $\pi_1,\dots,\pi_k$ of ropes where $\pi_1$ is the lower hull of $\calA$, $\pi_k$ is the upper hull of $\calA$, and consecutive ropes are obtained by flipping across an inner face.    The \emph{rope-length} of such a sweep is the maximum length (measured by the number of edges) among the used ropes, and the problem studied in this paper is to find a sweep that has small rope-length.

Graph $G_\calA$ (and generally any bipolar orientation) naturally gives rise to a dual graph $G_\calA^*$ that is also a bipolar orientation as follows.    Temporarily add an edge $(s,t)$ to $G_\calA$, and let $s^*,t^*$ be the two faces incident to it, with $s^*$ incident to the upper hull of $\calA$.  The vertices of $G_\calA^*$ are now $s^*,t^*$, and one vertex $F$ for each inner face of $G_\calA$.   For every edge $e=u\rightarrow v$ of $G_\calA$, let $F_\ell$ and $F_r$  be the faces that lie to the left and right when walking from $u$ to $v$.   (Since our edges are directed left-to-right, these faces are really above and below $e$, but ``left''/``right'' is the established term in the literature.)   We add to $G_\calA^*$ the \emph{dual edge $e^*$ of $e$}, which is $F_\ell\rightarrow F_r$.   
%\todo[inline,caption={}]{(Apr26) I didn't notice before, but Section 3 and 5 use the reverse direction of this for the dual graph.   Oops.    \\
%I will fix Section 3 and 5 for now, but longer-term it would make more
%sense to reverse the edge-directions of the dual (or to sweep from top
%the bottom).   Because the dual graph naturally expressed ``face $F$
%needs to be flipped before face $F'$''---but only if we reverse it.
%Unfortunately this means updating all figures and checking
%before/after/left/right/first/last in many places, so too late to do
%for the CCCG submission I think.\\
%G: Bottom-to-top is better for us: The primal sweep is a
%``leftmost-greedy'' sweep, and similarly the dual sweep is a
%``bottommost-greedy'' sweep. (Are these GREEDY connections
%mentioned?).
%Anyway, whatever convention we choose, it will be inconsistent: The
%dual of the dual of a bipolar orientation is never the original, but its reverse. So, choose our
%convention for our convenience.\\
%Thirdly, when speaking of left/right of $\pi^*$ in Definition 1, it is
%more pleasant when $\pi^*$ runs from bottom to top.
%\\I am not sure if the top-to-bottom
%numbering of the pseudolines ought to be changed as well.}
\todo[inline,caption={},color=lightgray]{(Jul 3) I had hoped to change the orientation of the dual to its reverse (see latex for various discussions why this is better), but there probably is no time for that since it involves redrawing many pictures.   Punt to the journal-version if we ever write one.}
Note that $e$ lies on the top chain of $F_r$ and the bottom chain of $F_\ell$, so in any sweep we must have swept $F_r$ \emph{before} we can sweep $F_\ell$. 
We think of dual graph $G_\calA^*$ as drawn such that each vertex $F$ is placed in the corresponding face of $G_\calA$, and each edge $e^*$ crosses the edge $e$ that it is dual to.   By definition, $e^*$ crosses $e$ from left to right.

\begin{figure}[ht]
\centering%
\includegraphics[page=2,scale=1.2%,width=\linewidth
]{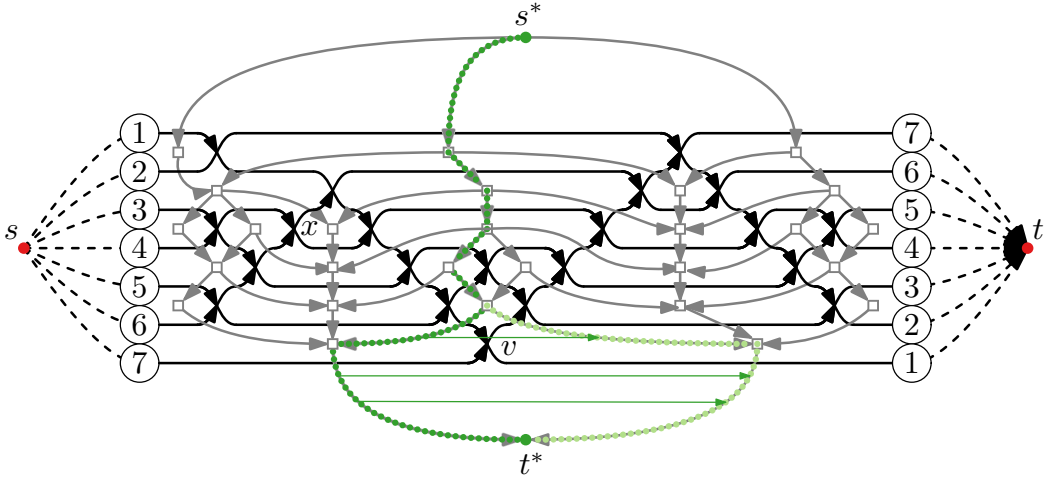}%
\caption{The dual graph $G_\calA^*$ with a dual rope $\pi^*$ (green dotted) that can be flipped across vertex $v$.}
\label{fig:dualGraph}
\end{figure}

Since $G_\calA^*$ is also a bipolar orientation, concepts such as ``rope'' and ``flipping across a face'' can also be applied to $G_\calA^*$.   For ease of distinction, we use the term \emph{dual rope} for a rope in $G_\calA^*$, and \emph{flipping across a vertex (of $G_\calA$)} for the operation of flipping across a face of $G_\calA^*$.
Note that any dual rope $\pi^*$ defines an $st$-cut by virtue of taking the edges 
of $G_\calA$ that it \emph{crossed} (i.e., whose duals it contained), and
symmetrically every rope $\pi$ defines an $s^*t^*$-cut.   Both these cuts are
\emph{directed}, i.e., contain only edges directed from the source-side to the
sink-side.

\section{A lower bound}
\label{sec:lb}

\begin{figure}[htb]
\centering
\includegraphics[scale=1.1,%width=0.99\linewidth,
page=6,trim=30 0 30 0,clip]{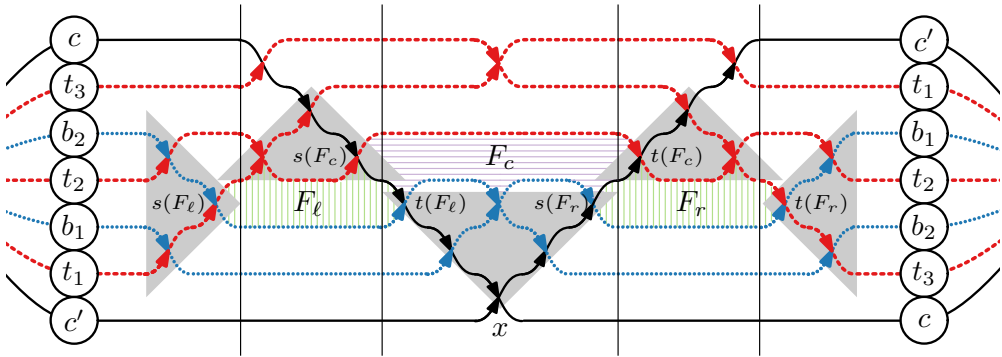}
\caption{Construction for the lower bound for $n=7$ ($K=1$); we
  need rope-length 11.
   The four vertical lines cut the arrangement
   into five % consecutive
   sections.
   %, which will be used in the proof of
   %Proposition~\ref{prop:proof-example}.
 }
\label{fig:lowerbound}
\label{fig:lowerbound7}
\end{figure}

\begin{theorem}
\label{thm:lower}
For $n=3 \bmod 4$, % and $n\geq 7$,
there exists
a pseudoline arrangement $\calA$ of $n$ $x$-monotone curves 
such that any sweep requires rope-length at least $\tfrac{7}{4}n-\tfrac{5}{4}$.
\end{theorem}
\begin{proof}
The construction is symmetric, and we describe it from left to right,
see \autoref {fig:lowerbound7} for the construction for $n=7$ and \autoref{fig:lowerbound15} 
%(in the appendix) %for the construction 
for $n=15$. Start with two curves $c,c'$
(black solid) that are at the top and bottom at the far left and
intersect in some point~$x$.   All other curves will pass above $x$.
Set $K=\tfrac{n-3}{4}$.
Between $c$ and $c'$ at the far left are $2K+1$ ``top'' curves (red, dashed)
at even positions, and $2K$ ``bottom'' curves (blue, dotted) at odd
positions.

\begin{figure*}[t]
\hbox to \linewidth{\hss  \includegraphics[scale=0.64,
  page=7]{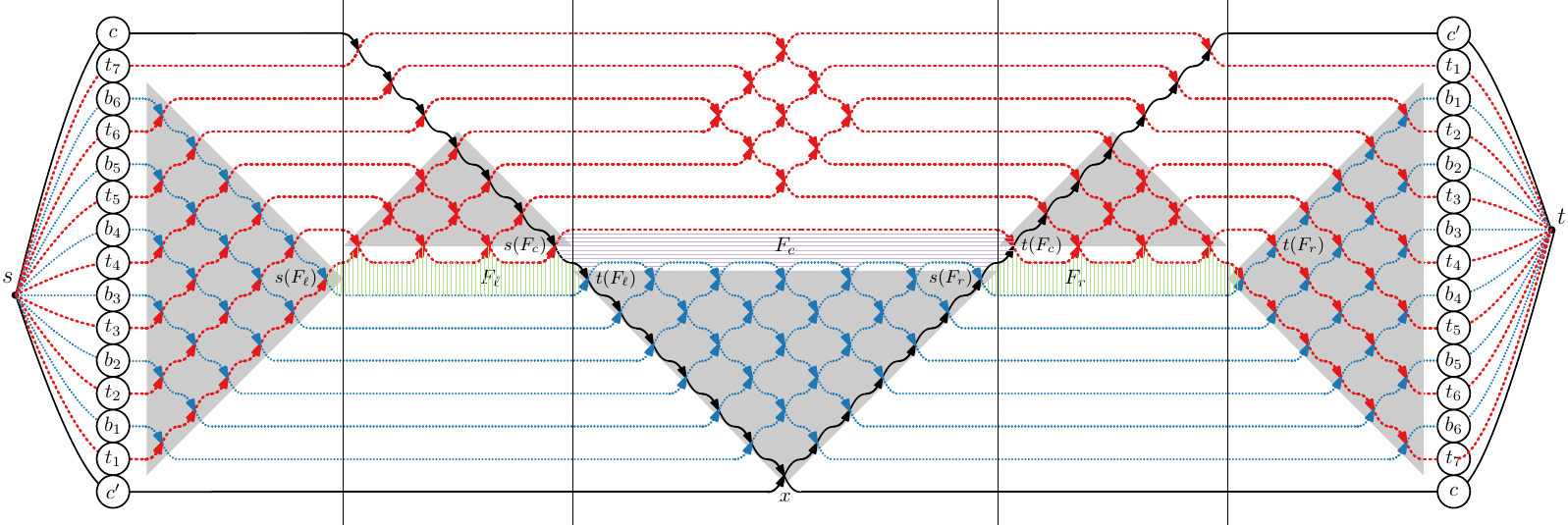}\hss}
% stick out the margin a little.
\caption{The lower-bound construction with $n=15$ pseudolines ($K=3$);
  we need rope-length 25.}
\label{fig:lower_bound_15}
\label{fig:lowerbound15}
\end{figure*}

\begin{figure*}[t]
  \centering
\includegraphics{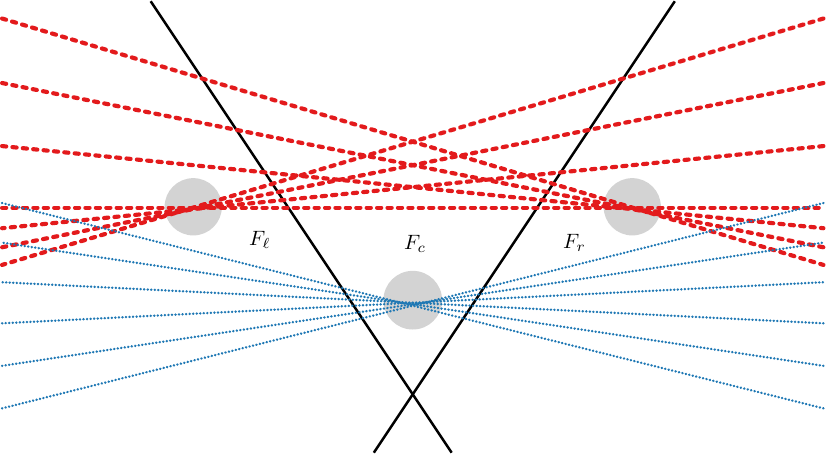}
\caption{
The lower-bound example for $n=15$ as an arrangement of straight lines.
The slopes of the seven red (dashed) and six blue (dotted) lines are evenly spaced, with red and
blue slopes interleaving. This ensures the appropriate intersection pattern
when the lines are extended far enough to the left and right.
%\par
In the three shaded disks, the lines are slightly perturbed from a common intersection point
so that they %lines
become incident to $F_\ell$,
$F_c$, and $F_r$, respectively.
}
  \label{fig:stretchable}
\end{figure*}

The arrangement consists of five consecutive sections, as indicated by
the vertical lines in the figure.
In the first section, the red curves move up and the blue curves move down
until they are separated, forming a $2K\times 2K$ half-grid
(shown shaded in Figures \ref {fig:lowerbound} and~\ref {fig:lowerbound15}).
So far there are no intersections between curves of the
same color. In the area below all red curves and above all blue
curves, there are three faces $F_\ell$, $F_c$, and $F_r$, separated from
each other by $c$ and $c'$.

Before the $2K{+}1$ red curves cross $c$, we let the lower $K{+}1$ of 
them cross each other in such a way that they all become incident to the top chain of
$F_\ell$. These curves, together with $c$, hence create a
$(K{+}1)\times (K{+}1)$ half-grid,
which forms the second section. (In terms of sorting networks, 
this half-grid is the \emph{bubble-sort} network.)

In the middle section, in the area above $x$, we do two things:
a) We cross the blue curves in such a way that they all become incident
to the bottom chain of $F_c$, forming a $(2K{+}1)\times(2K{+}1)$ half-grid together
with $c$ and $c'$.
b) We cross the upper $K$ with the lower $K$ red curves (the middle red curve 
remains uncrossed, as it meets all other red curves in the half-grids
above $F_\ell$ and $F_r$).

The right part of the construction is symmetric.
As shown in Figure~\ref{fig:stretchable},
this arrangement can even be drawn with straight lines.
Observe the following properties of $x$-monotone paths in the construction:
\begin{itemize}
\item Any $x$-monotone path from $s$ to the source $s(F_\ell)$ of $F_\ell$
	has length at least $2K$.
	This holds because such a path must traverse the $2K\times 2K$ half-grid, plus the edge from $s$ to reach the half-grid.
\item Any $x$-monotone path $\pi$ from $t(F_c)$ to $t$ has length at least $2K+1$.
This is obvious if $\pi$ walks along $c'$ until the intersection with the last red curve (and from there to $t$). 
So assume that it walks along $c'$ for $i<2K$ edges and then turns onto a red curve that brings us (perhaps after some more edges) to the half-grid right of $t(F_r)$.     It then 
	traverses a $(2K{-}i)\times (2K{-}i)$  half-grid, which takes $2K{-}i$
	edges, plus one more edge to $t$.  Hence the path
	has length at least $2K{+}1$.

\item Any $x$-monotone path from $t(F_\ell)$ to $s(F_r)$ has length at least $2K$, because it must go across the
$(2K{+}1)\times (2K{+}1)$ half-grid below $F_c$ and can (at best) 
	use shortcuts along the bottom chain of $F_c$.   
\end{itemize} 

Now we come to the actual proof.   Consider any sweep of $\calA$.   Since
the dual graph has edges $F_c\rightarrow F_\ell$ and $F_c\rightarrow F_r$, we must 
flip across both $F_\ell$ and $F_r$ before flipping across $F_c$.   By symmetry we
may 
assume that we flip across $F_r$  first, and consider the rope $\pi$ immediately after 
we flipped across $F_\ell$.   Then $\pi$ goes from $s$ to $s(F_\ell)$, from there along
the top chain of $F_\ell$ to $t(F_\ell)$, from there to $s(F_r)$ and $t(F_c)$ 
(since we have flipped across $F_r$ but not $F_c$ yet), and from there to $t$.
So
\begin{eqnarray*}
|\pi|& = &
|\pi(s,s(F_\ell))| % from s to source of F_ell 
+\text{length of top chain of $F_\ell$} \\ % along top of F_ell
&& +|\pi(t(F_\ell),s(F_r))| % most of bottom of F_c
%+(\text{$s(F_r){\rightarrow}t(F_c)$}) 
+1% from s(F_r) to t(F_c)
+|\pi(t(F_c),t)|\\ % from sink of F_c to sink
& \geq & 2K+K{+}2+2K+1+2K{+}1=7K+4
\end{eqnarray*}
which is at least $7\tfrac{n-3}{4}+4=\tfrac{7}{4}n-\tfrac{5}{4}.$
\end{proof}

The lower bound that we have proved is tight for these instances:
The primal-dual sweep  that we present in the next section
% will actually
achieves ropelength $(7n-5)/4=7K+4$: %in these instances:

\begin{proposition}
\label  {prop:proof-example}
The pseudoline arrangement $\calA$ of
\autoref{thm:lower}, consisting of
$n$ $x$-monotone curves, is swept by the primal-dual sweep of
Section~\ref{sec:up-sweep}
with
rope-length $\tfrac{7}{4}n-\tfrac{5}{4}$.
\end{proposition}
We give the proof in Section \ref{proof-example}

\section{An upper bound: The 
coordinated  % TB: removed to save a line
  primal-dual sweep}
\label{sec:up-sweep}

We now show an upper bound on the required rope-length by defining a sequence of ropes in $G_\calA$ and simultaneously a sequence of dual ropes that ``hug'' the ropes.
%\therese{Someone in the old writeup suggested ``hug''.   I liked it then, but don't particularly like it now; if you want to change it then do.  (Some other ideas:  ``move in tandem'', ``run parallel'').%  "stay close to each other"? }
To define this, we first need
a few other definitions and observations 
about a rope $\pi$ and a dual rope $\pi^*$
(see also Figure~\ref{fig:ropes-hug}).

\begin{figure}[ht]%
\centering
\includegraphics[page=3,scale=1.2%,%width=0.99\linewidth%,
% trim=0 0 0 0,clip
  ]{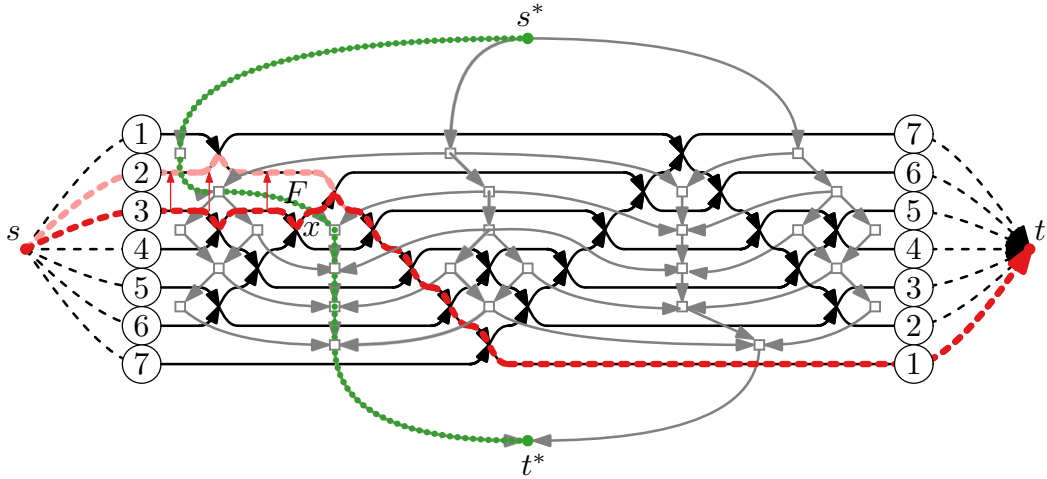}%\vspace*{-3mm}% TB: There is extra space between the figure and the caption, and I don't know of anything better than this hack to reduce it
\caption{A rope and a dual rope that hug each other. We can flip
across face $F$, which is to the left of the active edge.}
\label{fig:ropes-hug}
\end{figure}

Rope $\pi$ connects $s$ to $t$, hence must go across the directed $st$-cut defined by $\pi^*$, and can do so only once since $\pi$ is directed.   It follows that exactly one edge $e$ of $\pi$ is crossed by $\pi^*$; we call $e$ the \emph{active edge} and let $x$ be the point where it is crossed by $\pi^*$.
%\therese{I'm not super-fond of ``active'', if you aren't either then propose a better name. G: I like it: this is where the action happens} 
This \emph{crossing-point} $x$ splits the rope into two parts $\pi(s,x)$ and $\pi(x,t)$, and likewise splits
the dual rope into $\pi^*(s^*,x)$ and $\pi^*(x,t^*)$, and the properties that we
require will depend on which part we are in.

\begin{definition}   We say that a rope $\pi$ and dual rope $\pi^*$ \emph{hug each other}
if the following four (symmetric) conditions hold:   (1) for every edge $e$ in $\pi(s,x)$, the face to the left of $e$ belongs to $\pi^*$; (2) for every edge $e$ in $\pi(x,t)$, the face to the right of $e$ belongs to $\pi^*$; (3) for every edge $e^*$ in $\pi^*(s^*,x)$, the face of $G_\calA^*$ (hence vertex of $G_\calA$) to the right of $e^*$ belongs to $\pi$; (4) for every edge $e^*$ in $\pi^*(x,s^*)$, the vertex of $G_\calA$ to the left of $e^*$ belongs to $\pi$.  
%\therese{Mostly (3) and (4) follows from (1) and (2), but I had trouble proving that around $s$ and $t$, and so added it to the conditions instead} 
%\guenter{shouldn't it be ABOVE $e$ in (1)? $\pi$ goes from left to right!} \therese{see text around line 158:  Yes, that's the face above $e$, but traditionally this is called ``left of e'' since it is defined while walking along $e$.  }
\end{definition}

We will now define a sequence of \emph{rope pairs} (i.e., pairs of a rope $\pi$ and a 
dual rope $\pi^*$) such that the ropes sweep $G_\calA$, the dual ropes sweep 
$G_\calA^*$, and at all times $\pi$ and $\pi^*$ hug each other.    Then we argue that 
this implies rope-length at most $2n-2$ at all times.   We initialize
rope $\pi$ as the lower hull of $\calA$, so all edges of $\pi$ have $t^*$ to their right.
We initialize the dual rope $\pi^*$ to contain all faces incident to $s$, in order from top to bottom, so all edges of $\pi^*$ have $s$ to their right.
The active edge is the bottommost outgoing edge of $s$, 
and one easily verifies all conditions. 
(\autoref{sec:long-instance}
shows an example of a sweep from the beginning.)
To explain how to update the rope pair, we need some observations. 

\begin{claim}
\label{obs:topIncoming}
\label{cl:topIncoming}
(1) At any vertex $v\neq s$ of $\pi(s,x)$,  rope $\pi$ uses the top incoming edge. 
(2) At any vertex $v\neq t$ of $\pi(x,t)$,  rope $\pi$ uses the bottom outgoing edge. 
(3) At any face $F\neq s^*$ of $\pi^*(s^*,x)$,  dual rope $\pi^*$ crosses the first
edge of the top chain of~$F$.
(4) At any face $F\neq t^*$ of $\pi^*(x,t^*)$,  dual rope $\pi^*$ crosses the last 
edge of the bottom chain of~$F$.
\end{claim}
\begin{proof}
We only prove the first claim, the other three are symmetric.
Let $e$ be the incoming edge of $v$ on $\pi$, and
assume for contradiction that $e$ is not top incoming. 
Then the face $F$ to the left of $e$ is incident to two incoming edges of $v$,
hence $v=t(F)$ and $e$ is the last edge of the bottom chain of~$F$.
By the hugging-condition $F$ belongs to $\pi^*$; 
the next edge on $\pi^*$ hence crosses the bottom chain of~$F$.   But then $v=t(F)$ is on the $t$-side of the $st$-cut
defined by the dual rope $\pi^*$, contradicting that $v\in \pi(s,x)$.
\end{proof}

%\therese{A lot of the following claims could be simplified a bit if we assumed that no three pseudolines meet in a point and hence there is only a top incoming and a bottom incoming edge etc.    But I rather like the idea that this algorithm works for \emph{any} bipolar orientation (and the analysis works even if three pseudolines meet in a point), so I think it is worth keeping the occasionally slightly more complicated wording.}

\begin{claim}
\label{cl:canFlip}
Let $e$ be the active edge and let $v$ be its head and $F$ be the face to its left.   If $F\neq s^*$ or $v\neq t$, then we can flip $\pi$ across $F$ or flip $\pi^*$ across $v$, 
and the new pair of rope and dual rope hug each other.
\end{claim}
\begin{proof}  
The claim is illustrated in Figure~\ref{fig:can-flip}.
Assume first that $e$ is not top incoming, which implies 
that it is the last edge of the bottom chain of~$F$.   
We know that $F\neq s^*$ since all edges incident to $s^*$ are top incoming.
Since $\pi(s,x)$ only uses top incoming edges, $\pi$ must have traversed the entire bottom chain of $F$ and by $F\neq s^*$ we can hence flip across $F$ to get the new rope $\pi'$.   The new active edge is the first edge of the top chain of $F$ by Claim~\ref{obs:topIncoming}(3).   The hugging-conditions could be violated only at face $F$ (everywhere else the rope and dual rope are unchanged), and one easily verifies that they hold here because all new edges of $\pi'$ have $F$ to their right.

\begin{figure}[ht]
\centering
\begin{minipage}[t]{0.53\linewidth}%
\vspace*{0pt}
\includegraphics[page=4,scale=1.28,trim=12 67 215 40,clip]{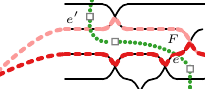}
\end{minipage}
\hspace*{\fill}
\begin{minipage}[t]{0.44\linewidth}
\vspace*{0pt}%
\includegraphics[page=5,scale=1.28,trim=88 30 160 70,clip]{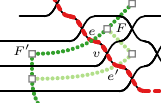}
\end{minipage}
\caption{Closeup of flipping across a face and a vertex.   Dual graph not shown.
}
\label{fig:can-flip}
\end{figure}

Now assume that $e$ is top incoming, which implies that $v\neq t$ since
otherwise $F=s^*$ and not both are allowed.    Let $F'$ be the face to the
right of $e$; this is in $\pi^*(x,t)$ since $e$ (as active edge) is crossed
by $\pi^*$.  All other incoming edges of $v$ 
are the last edge of the bottom chain of the faces to their left.
Applying Claim~\ref{obs:topIncoming}(4) repeatedly, starting with $F'\in \pi^*(x,t^*)$,
therefore dual rope $\pi^*$ must cross all incoming edges of $v$. 
So by $v\neq t$ we can flip the dual rope across $v$. 
By Claim~\ref{obs:topIncoming}(2) rope $\pi$ continues from $v$ along the bottom 
outgoing edge, which hence becomes the new active edge.    Again one easily
verifies the hugging condition, since all new edges of the new dual rope have
$v$ to their right.
\end{proof}

We hence update $\pi$ and $\pi^*$ as follows.  Let $e$ be the active edge, and let $v$ be its head and $F$ be the face to its left.    
If $F=s^*$ and $v=t$ then $e$ is the last edge of the upper hull. 
By Claim~\ref{obs:topIncoming}(1) hence $\pi$ is the upper hull and the sweep is finished.   By Claim~\ref{obs:topIncoming}(4) $\pi^*$ crosses all incoming edges of $t$, and so the sweep of the dual is also finished.    Otherwise
(either $F\neq s^*$ or $v\neq t$) we perform one of the flips that exists by Claim~\ref{cl:canFlip} and repeat.
%\therese{We thought that always only one of the flips is possible.   That's half-true---it is sometimes possible to flip either face or vertex, but only the face-flip would preserve a hugging rope-pair.   Not worth writing up IMO.}

\subsection{Analysis}

The sweeping algorithm as described would actually work for any bipolar orientation.
We now show that if the bipolar orientation comes from a pseudoline arrangement $\calA$ 
of $x$-monotone curves, then the rope-length is at most $2n-2$ at all times.  
%\therese{As far as I can tell, $x$-monotonicity is not used at all, everything could be defined via properties of the bipolar orientation along.  Is this worth  mentioning?}%
Enumerate the pseudolines from top to bottom in the order of incidence with $s$ as
$c_1,\dots,c_n$.   The \emph{index} of an edge $e$ is the index of the pseudoline
that supports $e$, i.e., along which $e$ runs.    The following observation is trivial
(it holds since pseudolines intersect only once, so one can go above the other
only once), but will be crucial for counting vertices later.

\begin{observation}
\label{obs:i<j}
At any vertex $v\neq s,t$, the indices of incoming edges increase from 
top to bottom, while the indices of outgoing edges decrease from top to bottom.
\end{observation}

An \emph{encounter} of rope $\pi$ with pseudoline $c_i$ is a maximal sub-curve
$\pi(v,v')$ that belongs to $c_i$.    Note that $v,v'$ are necessarily
vertices, and possibly $v=v'$.

\begin{corollary}
\label{obs:indexIncreases}
While walking along $\pi(s,x)$, the index $i$ of the current edge of $\pi$
can only increase, and any pseudoline $c_j$ encountered at the next vertex $v$
satisfies $j\geq i$.
\end{corollary}
\begin{proof}
Rope $\pi$ enters along the top incoming edge of $v$, 
hence $i$ is the smallest index of a pseudoline incident to $v$.
So all pseudolines encountered at $v$ (including the one along which $\pi$ leaves)
cannot have smaller index. 
\end{proof}

\begin{claim}
\label{cl:contiguous}
While walking along $\pi(s,x)$, we encounter every pseudo-line at most once.
\end{claim}
\begin{proof}
Assume for contradiction that we encounter pseudoline $c_i$ at least twice.
At the end of the first encounter we hence have a vertex
$v$ with $v\in c_i\cap \pi(s,x)$, but $\pi$ continues
beyond $v$ along some pseudoline $c_j$ with $j\neq i$.   
If $j>i$, then the index throughout $\pi(v,x)$ is at least $j>i$,
and so we cannot encounter $c_i$ again.   So we must have $j<i$,
which means that the outgoing edge of $\pi$ at $v$ is \emph{below} the outgoing
edge along $c_i$ by Observation~\ref{obs:i<j}.   Therefore $c_i$ has entered
the $s^*$-side
of the $s^*t^*$-cut defined by $\pi$.   Since $\pi(s,x)$ always uses
top incoming edges, there are no edges from the $s^*$-side to $\pi(s,x)$,
and so $c_i$ cannot encounter $\pi(s,x)$ again.
\end{proof}

\begin{claim}
\label{cl:2n-2}
At any time during the sweep, rope $\pi$ has length at most $2n-2$.
\end{claim}
\begin{proof}
Assign to $s$ the pseudoline along which $\pi$ leaves, and
assign to every vertex $v\neq s$ on $\pi(s,x)$ the pseudoline $c$ that supports 
the bottom incoming edge $e$ at $v$.   This assigns every pseudoline at
most once, for $e$ was \emph{not} in $\pi(s,x)$ by Claim~\ref{cl:topIncoming},
and so $v$ is the beginning of the unique encounter of $c$ with $\pi(s,x)$.
(This also shows that $c$ was not assigned to $s$).
So $\pi(s,x)$ has at most $n$ vertices, and symmetrically
$\pi(s,t)$ has at most $n$ vertices and the rope-length is at most $2n-1$.

We claim that this is not tight.  Assume for
contradiction that at some point rope $\pi$ has length exactly $2n-1$,
so $\pi(s,x)$ has $n$ vertices and \emph{all} pseudolines have been
assigned to some vertex of $\pi(s,x)$.   Observe that $c_1$ must have been 
assigned to $s$,  for otherwise the index of $\pi(s,x)$ would be greater than 1 throughout, 
so $\pi(s,x)$ could not encounter $c_1$,  so $c_1$ would not be 
assigned to a vertex.     Also observe that $c_2$ must have been assigned
to a vertex $v$ that lies on $c_1$, because it is not assigned to $s$, and
we assign (by Observation~\ref{obs:i<j} and Claim~\ref{cl:topIncoming}(1))
a pseudoline $c_j$ to a vertex $v\neq s$ only if $\pi(s,x)$ has index less
than $j$ when it reaches $v$.  In particular therefore $c_1$ and $c_2$ intersect
at a point on $\pi(s,x)$.   By a completely symmetric argument, 
$c_1$ and $c_2$ intersect again at a point on $\pi(x,t)$.   
This is not possible in a pseudoline arrangement.
\end{proof}

\begin{theorem}
For every pseudoline arrangement of $n$ $x$-monotone curves, 
there exists a sweep with rope-length at most $2n-2$.
\end{theorem}

A few comments are in order.   First,
% as the example shown in the appendix illustrates,
the bound is tight: for
some arrangements, this particular method of computing a sweep requires
rope-length $2n-2$.
%An example is shown
\autoref{sec:long-instance} describes a family of examples where
the primal-dual sweep uses ropelength $2n-2$.

Also, our
coordinated primal-dual sweep can be interpreted as a \emph{left-first greedy} sweep: At each stage, the rope $\pi$ selects the leftmost possible position where it can flip over a face. The
dual rope $\pi^*$ can be interpreted as guiding the search for the sweep position:
As long as a flip is not possible at the current position of the active edge, the
active edge advances to the right, and this corresponds to a dual flip.
%Such a left-first greedy method was used in an algorithm by Alvarez and Seidel as a tool to count the number of triangulations \cite{AlvarezS13}.
% TB: Shortened this to save lines GR: Reinstated the longer version
In fact, this {left-first greedy} method was used by Alvarez and Seidel to sweep a rope over a (hypothetical) triangulation of a set of points, in their algorithm for counting the number of triangulations \cite{AlvarezS13}.

It is striking that the same procedure can be interpreted as a \emph{bottom-first greedy} dual sweep, where the primal rope $\pi$ plays the role of guiding the sweep of~$\pi^*$.

\section{Proof of  Proposition~\ref{prop:proof-example}}
\label{proof-example}

Now that we have defined the
primal-dual sweep, we can apply it to the lower-bound examples
of Section~\ref{sec:lb} and prove Proposition~\ref{prop:proof-example}.

 We use the particular drawing of the arrangement
in Figures~\ref{fig:lowerbound7} and~\ref{fig:lowerbound15}.
The vertices are drawn in $2K+(2K+1)
+(4K+1)+(2K+1)+2K=12K+3$ vertical layers, excluding $s$ and $t$,
as indicated by the vertical lines in the figure.
This gives an a-priori upper bound of $12K+4$ on the length of any
potential rope.
For the claimed bound $\tfrac{7}{4}n-\tfrac{5}{4}=7K+4$,
we need to show that the rope always skips at least $5K$ layers.

The greedy primal-dual sweep algorithm starts with flips in the left part until the rope follows line $c$.
During this time, the rope contains the long horizontal rightmost section of $c$, which
skips $6K+1$ layers, so we are on the safe side.
From now on, the rope will always contain the long first section of $c$, skipping over $2K$ layers. So we need to skip only $3K$ layers in the rest of the rope.

The algorithm will now flip over each face below $F_r$, followed by a ripple %of faces 
to the left until the lower
border of $F_c$ is reached. The rope becomes gradually longer and longer, but stays far from getting critical.
In the end of this phase, the rope runs along the lower edge of $F_c\cup F_r$, thereby skipping
$2K+(2K+1)$ layers.
(One step prior to this situation, the rope was one segment longer, and this is where the maximum ropelength so far was achieved.)
Now the rope flips over $F_r$ and gets to the critical situation conjured up in the lower-bound proof: Only $2K+K=3K$ layers are skipped.
Next, the rope flips over $F_c$. From now on, the rope contains the single long edge at the upper bounday of $F_c$, or two of the long edges on either side of the central top diamond,
skipping at least $4K+1$ layers.
This concludes the proof.
%QED.
%[In fact, the greedy sweeps $F_\ell$ and $F_r$ in the opposite order from our lower-bound argument; should we eventually swap them in that argument?] [Maybe this "more geometric" argument that refers to
%a particular drawing would also be useful for the lower-bound argument?]

%%%%%%%%%%%%%%%%%%%%%%%%%%%%%%%%%%%%%%%%%%%%%%%%%%%%%%%%%%%%%%%%%%%%%%%%
\newcommand{\indeg}{d^-}
\newcommand{\outdeg}{d^+}
\newcommand{\cut}[2]{\ensuremath{[#1{:}#2]}}

\section{NP-hardness}
\label{sec:hardness}

In this section, we reduce our sweep-problem to solving {\sc Directed Cutwidth}
in $G_\calA^*$.   Then we show that {\sc Directed Cutwidth} is NP-hard
even in planar graphs with maximum degree~6.   Unfortunately this does
not prove the sweep-problem NP-hard since the graph that we construct cannot
be the dual graph of a pseudo-line arrangement (it has vertices of degree~2 and many sources and sinks).

We need a few definitions.  
Fix a vertex order $\sigma=\langle v_1,\dots,v_n\rangle$ of $G$.  For $1\leq i\leq n$,
the \emph{$i$th cut} (or \emph{cut after $v_i$}) is the set of edges
$(v_h,v_j)$ with $h\leq i< j$.   
The maximum cardinality of these cuts 
is the \emph{width} of the vertex order, 
and the \emph{cutwidth} of graph $G$ is the minimum width over all vertex orders.

The cutwidth is defined for undirected graphs, but for directed acyclic graphs
there exists a natural restriction, apparently first studied in \cite{BFT09}:
% TB: BFT09 calls this ``directed cutwidth'', so I changed our name to match that
The  \emph{directed cutwidth} of a directed acyclic graph $G$ is the
minimum width of a vertex order of $G$ that is a \emph{topological order}, i.e.,
where every edge is directed from a lower-indexed to a larger-indexed vertex.

\begin{lemma}
Let $\calA$ be a pseudo-line arrangement with $x$-monotone curves.   Then
$\calA$ has a  sweep with rope-length at most $w$ if and only if $G_\calA^*$ has
directed cutwidth at most $w$.
\end{lemma}
\begin{proof}
We only show one direction, the other is similar.
Fix a sweep with rope-length $w$.   This defines a sequence $\sigma=\langle F_1,\dots,F_k\rangle$ of the inner faces of $G_\calA$ via the order in which the sweep flips the rope across faces.
We append $s^*=:F_{k+1}$ and pre-pend $t^*:=F_0$ to this sequence since the rope begins incident
to $t^*$ and ends incident to $s^*$.   Sequence $\sigma$ hence gives a vertex
order $F_0,F_1,\dots,F_{k+1}$ of $G_\calA^*$.    Any directed edge $F_\ell\rightarrow F_r$
of $G_\calA^*$ is dual to an edge $e$ of $G_\calA$ that is on the upper chain of $F_r$ and the lower chain of $F_\ell$. 
So the sweep must flip across $F_r$ \emph{before} flipping across $F_\ell$, i.e., $r<\ell$.   So in our face order all edges of $G_\calA^*$ are directed right-to-left, and reversing it (which does not affect the width) gives a topological order.    Finally the edges of the $i$th cut are dual to the edges of the rope after flipping across  $F_i$, and vice versa. Therefore the width of the topological order is the same as the rope-length. 
\end{proof}

So we are interested in the complexity of 
problem {\sc Directed Cutwidth}, %which is
the % obvious
decision version of
the problem: Given a directed acyclic graph $G$
and an integer $w$, is there a topological order of width at most $w$?
Surprisingly% enough
, the complexity of this problem does not appear to have been 
studied much in the literature.  
Wu et al.~\cite{WuAPL14} showed that {\sc Directed Cutwidth}
(not specifically named there, but appearing in row 6 of their Table~1) 
is SSE-hard to approximate (the constructed graphs are non-planar).  
There are also some positive results; in particular
{\sc Directed Cutwidth} has a linear-time algorithm if $w$ is a constant  \cite{BFT09},
and for series-parallel graphs it can be computed in quadratic time
\cite{BJT23}. 
But we have the following new result: 

\begin{theorem}
{\sc Directed Cutwidth} is NP-hard, even in planar graphs with
maximum degree~6.
\end{theorem}
\begin{proof}
The reduction is from {\sc Cutwidth}, 
which is known to be NP-hard, even for a planar graph with 
maximum degree 3 \cite{MonienS88}.  
So assume that we are given a planar graph $G$ with maximum degree 3 and
an integer $w$ and we want to test whether its cutwidth is at most~$w$. 
We may assume that $G$ has no isolated vertices or isolated edges:
They do not affect the
cutwidth, except in the trivial case that $G$ consists exclusively of isolated edges and vertices.
We create a directed graph $H$ as follows (see Figure~\ref{fig:GtoH}).  We retain all vertices
of~$G$, and replace every edge $e=(v,w)$ by a source $s_e$ and a sink $t_e$ that are both incident to both $v,w$. 
(A similar transformation, using only a sink, was used in \cite{WuAPL14}.)  

\begin{figure}[ht]
\centering
\includegraphics[%width=0.9\linewidth,page=1
]{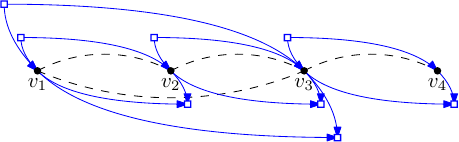}
\caption{From a vertex order of $G$ (black dashed) to a topological
order of $H$ (blue solid).   For ease of reading we offset sources to be
above and sinks to be below vertices of $G$.}
\label{fig:GtoH}
\end{figure}

We claim that $G$ has a vertex order $\sigma_G$ of width at most $w$ if
and only if $H$ has a topological order $\sigma_H$ of width at most
$2w+2$.
This implies the correctness of the reduction.
The remainder of the proof will show the claimed relation between the
widths of $G$ and~$H$  in both directions.

To convert $\sigma_G$ to $\sigma_H$, simply add (for each edge $e$ of $G$)
source $s_e$ just before the first endpoint of $e$ in $\sigma_G$, and
sink $t_e$ just after the second endpoint of $e$ in $\sigma_G$.
This doubles the degrees of all vertices
of $G$ (so the maximum degree of $H$ is 6).   Also the undirected version of
$H$ can be obtained by duplicating all edges of $G$ and then subdividing all
edges; in particular if $G$ is planar then so is $H$.

%Elementary arguments (using that $G$ has maximum degree 3)
We want to show that
$\sigma_H$ then has width at most $2w+2$.     To convert $\sigma_H$ to
$\sigma_G$, initially simply take the induced vertex order,   which is easily seen to 
have width at most $w+1$. 
This can be tight (say at the $i$th cut) only if
%The $i$th cut can have width $w+1$ only if
$v_i$ has no neighbours on the left while $v_{i+1}$ has no neighbours
on the right.   Call such a pair $(v_i,v_{i+1})$ \emph{improvable}:
exchanging the two vertices in the order improves the size of the
$i$th cut and leaves all other cuts after vertices unchanged.   
Exchanging all improvable pairs hence gives the desired~%
%vertex order of~$G$.
$\sigma_G$.

%\subsection{More details on ``NP-hardness''}

\newcommand{\Cutbefore}[2]{\ensuremath{C^{\shortleftarrow}_{#2}\!(#1)}}
\newcommand{\Cutafter}[2]{\ensuremath{C^{\shortrightarrow}_{#2}\!(#1)}}
\newcommand{\cutafter}[2]{|\Cutafter{#1}{#2}|}
\newcommand{\cutbefore}[2]{|\Cutbefore{#1}{#2}|}

%In this section, we fill in the details of the NP-hardness proof 
%of Section~\ref{sec:hardness}. 
%Recall given a graph $G$, we created the directed acyclic graph $H$ by replacing
%every edge $e$ of $G$ by a source $s_e$ and a sink $t_e$ that both are
%adjacent to both endpoints of $e$.

To argue about the relation between the widths,
we need some notation. For any vertex order $v_1,\dots,v_n$ of $G$,
and any $i=1,\dots,n$, write $L_i$ [$R_i$] for the set of edges in $G$
that are incident to $v_i$ and whose other endpoint is left [right] of $v_i$ 
in the vertex order.   
Also, let $B_i$ be the set of edges that \emph{bypass} $v_i$, i.e., have
the form $(v_h,v_j)$ for $h<i<j$, and note that the cut before and after
$v_i$ have size $|B_i|+|L_i|$ and $|B_i|+|R_i|$, respectively.

For both $G$ and $H$, we write $\Cutbefore{v}{}$ and
$\Cutafter{v}{}$ for the cuts directly before and after a vertex~$v$, 
respectively, and indicate with a subscript which graph this applies to.
(The vertex order will be clear from context.)

\begin{claim}
If $G$ has a vertex order $v_1,\dots,v_n$ of width $w$, then $H$ has a 
topological order $\sigma_H$ of width at most $2w+2$.
\end{claim}
\begin{proof}
As sketched earlier, $\sigma_H$ is obtained by inserting, for each edge $e$, the
source just before the left end of $e$ and the sink just after the right end of $e$. 
Put differently, for $i=1,\dots,n$, list all sources of edges of $R_i$
(in arbitrary order), then list $v_i$, the list all sinks of edges in $L_i$
and proceed to the next $i$.  See Figure~\ref{fig:GtoH}
for an example, and verify that we indeed obtain a topological order.
Also notice that scanning $\sigma_H$ from left to right, the cut-sizes increase 
when we pass a source and decrease when we pass a sink, 
so the maximize cut-size of $\sigma_H$ must occur immediately before or
after some original vertex $v_i$ of $G$.   

One verifies that $\Cutbefore{v_i}{H}$ contains exactly two edges each for each edge
in $L_i\cup B_i\cup R_i$, due to edges in $\Cutbefore{v_i}{G}$ and sources for edges
in $R_i$, respectively.     Therefore $\cutbefore{v_i}{H}=
2(|L_i|+|B_i|+|R_i|)$, and by symmetry, 
this is also equal to %the same bound holds for
$\cutafter{v_i}{H}$.  Since 
$$|B_i|+\max\{|L_i|,|R_i|\}
= \max\{\cutbefore{v_i}{G},\cutafter{v_i}{G}\}\leq w,$$ 
the width of $\sigma_H$ is at most
$2(|B_i|+\max\{|L_i|,|R_i|\}+\min\{|L_i|,|R_i|\})\le
2w+2\min\{|L_i|,|R_i|\}\leq 2w+2$ since $|L_i|+|R_i|\leq \deg_G(v_i)\leq 3$.
\end{proof}

For the other direction, we must convert a topological order of $H$
into a vertex order of $G$ of small width.
Recall that in a vertex order of $G$, the pair $(v_i,v_{i+1})$ (for some $1\leq i<n$) 
is called an \emph{improvable pair} if $L_i=\emptyset=R_{i+1}$, see also Figure~\ref{fig:HtoG}.

\begin{claim}
If $H$ has a topological order $\sigma_H$ of width $2w+2$, then in the induced
vertex order $v_1,\dots,v_n$ of $G$, the $i$th cut has width at most $w{+}1$
for all $i<n$, and equality holds only if $(v_i,v_{i+1})$ is an improvable pair.
\end{claim}
\begin{proof}
%This holds for $i=n$ since the last cut is always empty, so assume that $i<n$.
We have to bound $|B_i|+|R_i|$, and will show that all these edges, \emph{and}
the edges of $L_i$, had contributed to $\Cutafter{v_i}{H}$, so there cannot
be too many of them.   For $e\in L_i \cup B_i\cup R_i$, the left end was $v_i$ or
farther left, while the right end was $v_{i}$ or farther right.   Since
$\sigma_H$ is a topological order, 
source $s_e$ was strictly before $v_i$ and sink $t_e$ was strictly after $v_{i}$ in $\sigma_H$,
and so in $\sigma_H$ this contributed two edges to $\Cutafter{v_i}{H}$.
%Consider an edge $e$ in $\Cutafter{v_i}{G}$, so $e=(v_h,v_j)$ for $h\leq i < j$.
%Source $s_e$ is before $v_h$ and sink $t_e$ is after $v_j$ in $\sigma_H$
%since we have a topological order.   Therefore $(s_e,v_j)$ and $(v_h,t_e)$
%both belong to all cuts in the range between $v_i$ and $v_{i+1}$.
%For any edge $f\in L_i$, say $f=(v_g,v_i)$ for some $g<i$, sink $t_f$ comes after $v_{i}$
%in $\sigma_H$ and so belongs to $|\Cutafter{v_i}{H}|$, but was not counted earlier.
Therefore 
%$$2w+2\geq \cutafter{v_i}{H} = |L_i|+2|B_i|+2|R_i|,$$
%\guenter{fixed from $\geq2|L_i|$ to $=|L_i|$.}%
%\therese{What is wrong with $2|L_i|$?   They do contribute two edges to the cut after $v_i$.  Also, equality does not necessarily hold.   I will undo the change.}
$$2w+2\geq \cutafter{v_i}{H} \geq  2|L_i|+2|B_i|+2|R_i|,$$
which implies that $\cutafter{v_i}{G}=|B_i|+|R_i|\leq w+1$ and
equality can hold only if $L_i=\emptyset$.
Symmetrically arguing via the cut before $v_{i+1}$ in $\sigma_H$,
one sees that
%\begin{align*}
%2w+2& \geq 2\cutbefore{v_{i+1}}{H} \geq 2(|L_{i+1}|+|B_{i+1}|+|R_{i+1}|) \\
%& \geq 2(\cutbefore{v_{i+1}}{G}+|R_{i+1}|)
%= 2(\cutafter{v_{i}}{G}+|R_{i+1}|),
%\end{align*}
\begin{align*}
2w+2& \geq 2\cutbefore{v_{i+1}}{H} \geq 2(|L_{i+1}|+|B_{i+1}|+|R_{i+1}|) 
\end{align*}
and so 
$\cutafter{v_i}{G}=\cutbefore{v_{i+1}}{G}=|L_{i+1}|+|B_{i+1}|\leq w{+}1$ and
equality can only hold if also $R_{i+1}=\emptyset$.
\end{proof}

\begin{figure}[ht]
\centering
\includegraphics[%width=0.9\linewidth,
page=2]{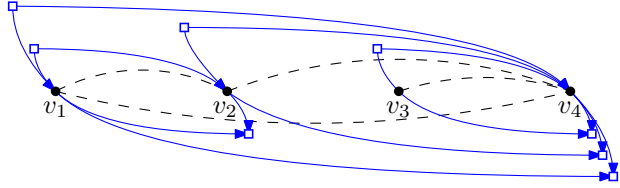}
\caption{From a topological
  order of $H$ (blue solid)
of width $2w+2=6$
  to a vertex order of $G$ (dashed black),
but it may not have optimal width:   $G$ has cutwidth $w=2$ (see 
Figure~\ref{fig:GtoH}), but the cut between $v_3$ and $v_4$ has
width 3.   Note that $L_3=\emptyset=R_4$, i.e., $(v_3,v_4)$ is improvable.}
\label{fig:HtoG}
\end{figure}

Figure~\ref{fig:HtoG} shows an example where the width of the induced
vertex order $\sigma_G$ is indeed $w+1$.
So we are not done yet with the reverse direction of the reduction.   But
observe that if the pair $(v_i,v_{i+1})$ is improvable, then by exchanging their
order
all edges in $R_i$ and $L_{i+1}$
are removed from the cut between them,
except the edge $v_iv_{i+1}$ if it exists.
Since we have excluded the cases that
$v_i$ or $v_{i+1}$ are isolated vertices or
$v_iv_{i+1}$ is an isolated edge, the cut strictly improves.
%In these exceptional cases, the graph is disconnected, and we could
%have treated the connected components separately.
%is improvable, then by exchanging their
%and there are such edges since $v_i,v_{i+1}$ are not
%isolated vertices but $L_i=\emptyset=R_{i+1}$.
All other cuts remain
unchanged.
We repeat this until no improvable pair remains.
In the end,
all cut-sizes are at most $w$ as desired, and $G$ has cutwidth
at most $w$.
\end{proof}

\section{Computer experiments}

We ran some computer experiments, exhaustively trying all
pseudoline arrangements with up to $n=9$ pseudolines.
(We used a \textsc{Python} version of the arrangement enumeration algorithm in~\cite{Rote25}.)
Each arrangement was subjected to a rather brute-force attack to
find the shortest rope-length, by essentially looking for a path in the graph
whose nodes represent all possible ropes.
The data that we found are displayed in
\autoref{tab:number}. For $n$ of the form $n=4k+3$, the results on
the maximum agree
with the lower bound of \autoref{thm:lower}.

\begin{table}[htb]
  \centering
  \noindent
  \vbox{
\def\gobble#1{}
\halign{\strut\hfil$#$
   &\hskip-0em\hfil$#$
   &\hskip-0em\hfil$#$
  &\gobble{#}%\ \hfil$#$
  &\gobble{#}%\ \hfil$#$
  &\gobble{#}%\ \hfil$#$
  &\gobble{#}%\ \hfil$#$
  &\gobble{#}%\ \hfil$#$
  &\hskip-0em\hfil$#$\cr
%&\hbox{[\href{https://oeis.org/A006247}{A006247}]}
%&\hbox{[\href{https://oeis.org/A063666}{A063666}]}
%&\Delta = {}&&&\hbox{[\href{https://oeis.org/A006245}{A006245}]}
% \cr 
  n &
 \hbox{min}&
 \hbox{max}&
  \hbox{\#AOT} & \hbox{\#OT}
 & \hbox{\#nonr.\ %ealizable
   AOT}
 %\Delta
 &
 &
 & \hbox{\#%$x$-monotone
   PSLA}
\cr
\noalign{\hrule}
2 & 2 & 2 & 1 & 1 & 0 & 0 & 1 & 1 \cr
3 & 4 & 4 & 1 & 1 & 0 & 0 & 1 & 2 \cr
4 & 5 & 5 & 2 & 2 & 0 & 0 & 2 & 8 \cr
5 & 6 & 7 & 3 & 3 & 0 & 0 & 3 & 62 \cr
6 & 7 & 9 & 16 & 16 & 0 & 0 & 12 & 908 \cr
7 & 8 & 11 & 135 & 135 & 0 & 0 & 28 & 24{,}698 \cr
8 & 9 & 12 & 3{,}315 & 3{,}315 & 0 & 0 & 225 & 1{,}232{,}944 \cr
9 & 10 & 14 & 158{,}830 & 158{,}817 & 13 & 0{,}01\,\% & 825 &\ 112{,}018{,}190 \cr
%10 & 14{,}320{,}182 & 14{,}309{,}547 & 10{,}635 & 0{,}07\,\% & 13{,}103 & 18{,}410{,}581{,}880 \cr
%11 & 2{,}343{,}203{,}071 & 2{,}334{,}512{,}907 & 8{,}690{,}164 & 0{,}37\,\% & 76{,}188 & 5{,}449{,}192{,}389{,}984 \cr
%12 & 691{,}470{,}685{,}682 &  &  &  &  &2{,}894{,}710{,}651{,}370{,}536 \cr
%13 & 366{,}477{,}801{,}792{,}538
%&  &  &  &  &2{,}752{,}596{,}959{,}306{,}389{,}652\cr
%4675651520558571537540
}}

\smallskip
\caption{min/max:
  The shortest and longest rope-length required for
    pseudoline arrangements with $n$ pseudolines.
  \#PSLA: the number of combinatorial types of $x$-monotone
  pseudoline arrangements with $n$ pseudolines
  %\hbox
  (sequence
    \href{https://oeis.org/A006245}{A006245}
in
the
Online Encyclopedia of Integer
Sequences% \cite{OEIS}.
  ).
  }
  \label{tab:number}
\end{table}

The lower bound is apparently $n+1$, except for $n=2$.
The number of % extremal
arrangements
that require the maximum rope-length
% with larger size
grows very quickly.
For example, among
the arrangements of 7 pseudolines, there are exactly two that require
rope-length 11, up to symmetries.
On the other hand, with 8 pseudolines, 
1184 arrangements among the
1,232,944 arrangements need rope-length 12.
%\guenter{I don't remember: Were our lower-bound constructions inspired by looking at those examples?} \therese{I don't remember either.  Does it matter?} 

\section{Summary and outlook} 

 We have studied the problem of sweeping a pseudoline arrangement of $n$ $x$-monotone curves using a rope between the points of infinity.  The only permitted move is to flip parts of the rope from the bottom chain to the top chain of a face, and the goal is to keep the number of edges on the rope small.   We argue that the worst-case rope-length is in $\Theta(n)$, and specifically, at most $2n-2$ (for all arrangements) and at least $\tfrac{7}{4}n-\tfrac{5}{4}$ (for some arrangements).   

The most tantalizing open problem is the complexity of finding the shortest
rope, possibly for an arbitrary bipolar orientation instead of a
pseudoline arrangement. %NP-hardness.  
We proved NP-hardness of {\sc DirectedCutwidth}, which is closely related to our
problem via duality.    But the graph that we construct for the NP-hardness
%has vertices of degree 2 and 
ihas many sources and sinks, and so 
is not the dual graph of a pseudoline arrangement,
and proving NP-hardness of the original problem or finding a polynomial-time
algorithm for it remains open.

%As for other open problems, 
Weighted versions could also be of interest, for example if
edge-weights are the edge-lengths in a straight-line drawing of %graph
$G_\calA$.

%\therese{This last open problem is inspired by something that Bruce asked; ``What if these are lines in a unit square and we want to minimize the (Euclidean) rope length? Is a rope of constant length always sufficient?'' I have absolutely no idea how to tackle it but it doesn't hurt to ask.}
\changed{Our sweep by definition is monotone in the sense that every inner face is swept exactly once.   Could a shorter rope-length ever be achieved if we are permitted to reverse some flips?
As discussed at the end of the introduction, we suspect that (as for the homotopy height under some restrictions on the input \cite{CCMOR21}) repeatedly sweeping a face cannot shorten the rope-length, but this remains open.
}
\todo{added open problem; Erin please read}
%One can argue that each path in the sequence can be assumed to be weakly simple, and under some restrictions on the input the sequence of paths is monotone in the sense that every face is swept exactly once \cite{CCMOR21}.

\bibliographystyle{plainurl} 
\bibliography{ref}

\begin{thebibliography}{10}

\bibitem{AlvarezS13}
Victor Alvarez and Raimund Seidel.
\newblock A simple aggregative algorithm for counting triangulations of planar
  point sets and related problems.
\newblock In Guilherme~Dias da~Fonseca, Thomas Lewiner, Luis~Mariano
  Pe{\~{n}}aranda, Timothy~M. Chan, and Rolf Klein, editors, {\em Symposium on
  Computational Geometry 2013, SoCG '13, Rio de Janeiro, Brazil, June 17--20,
  2013}, pages 1--8. {ACM}, 2013.
\newblock \href {https://doi.org/10.1145/2462356.2462392}
  {\path{doi:10.1145/2462356.2462392}}.

\bibitem{BRSSS04}
J{\'{o}}zsef Balogh, Oded Regev, Clifford~D. Smyth, William~L. Steiger, and
  Mario Szegedy.
\newblock Long monotone paths in line arrangements.
\newblock {\em Discret. Comput. Geom.}, 32(2):167--176, 2004.
\newblock \href {https://doi.org/10.1007/S00454-004-1119-1}
  {\path{doi:10.1007/S00454-004-1119-1}}.

\bibitem{BCEMO19}
T.~Biedl, E.~Chambers, D.~Eppstein, A.~de~Mesmay, and T.~Ophelders\student{}.
\newblock Homotopy height, grid-major height and graph-drawing height.
\newblock In Daniel Archambault and Csaba~D. T{\'{o}}th, editors, {\em Graph
  Drawing and Network Visualization ({GD} 2019)}, volume 11904 of {\em Lecture
  Notes in Computer Science}, pages 468--481. Springer, 2019.
\newblock \href {https://doi.org/10.1007/978-3-030-35802-0\_36}
  {\path{doi:10.1007/978-3-030-35802-0\_36}}.

\bibitem{BFT09}
Hans~L. Bodlaender, Michael~R. Fellows, and Dimitrios~M. Thilikos.
\newblock Derivation of algorithms for cutwidth and related graph layout
  parameters.
\newblock {\em J. Comput. Syst. Sci.}, 75(4):231--244, 2009.
\newblock \href {https://doi.org/10.1016/J.JCSS.2008.10.003}
  {\path{doi:10.1016/J.JCSS.2008.10.003}}.

\bibitem{BJT23}
Hans~L. Bodlaender, Lars Jaffke, and Jan~Arne Telle.
\newblock Typical sequences revisited - computing width parameters of graphs.
\newblock {\em Theory Comput. Syst.}, 67(1):52--88, 2023.
\newblock \href {https://doi.org/10.1007/S00224-021-10030-3}
  {\path{doi:10.1007/S00224-021-10030-3}}.

\bibitem{Brightwell2009}
Graham~R. Brightwell and Peter Winkler.
\newblock Submodular percolation.
\newblock {\em SIAM Journal on Discrete Mathematics}, 23(3):1149–1178,
  January 2009.
\newblock \href {https://doi.org/10.1137/07069078x}
  {\path{doi:10.1137/07069078x}}.

\bibitem{CCMOR21}
Erin~Wolf Chambers, Gregory~R. Chambers, Arnaud de~Mesmay, Tim Ophelders, and
  Regina Rotman.
\newblock Constructing monotone homotopies and sweepouts.
\newblock {\em Journal of Differential Geometry}, 119(3):383--401, 2021.
\newblock \href {https://doi.org/10.4310/jdg/1635368350}
  {\path{doi:10.4310/jdg/1635368350}}.

\bibitem{CMO18}
Erin~Wolf Chambers, Arnaud de~Mesmay, and Tim Ophelders.
\newblock On the complexity of optimal homotopies.
\newblock In Artur Czumaj, editor, {\em Proceedings of the Twenty-Ninth Annual
  {ACM-SIAM} Symposium on Discrete Algorithms, {SODA} 2018, New Orleans, LA,
  USA, January 7-10, 2018}, pages 1121--1134. {SIAM}, 2018.
\newblock \href {https://doi.org/10.1137/1.9781611975031.73}
  {\path{doi:10.1137/1.9781611975031.73}}.

\bibitem{CN98}
M.~Chrobak and S.~Nakano.
\newblock Minimum-width grid drawings of plane graphs.
\newblock {\em Comput. Geom.}, 11(1):29--54, 1998.
\newblock \href {https://doi.org/10.1016/S0925-7721(98)00016-9}
  {\path{doi:10.1016/S0925-7721(98)00016-9}}.

\bibitem{FOR95}
Hubert de~Fraysseix, Patrice~Ossona de~Mendez, and Pierre Rosenstiehl.
\newblock Bipolar orientations revisited.
\newblock {\em Discret. Appl. Math.}, 56(2-3):157--179, 1995.
\newblock \href {https://doi.org/10.1016/0166-218X(94)00085-R}
  {\path{doi:10.1016/0166-218X(94)00085-R}}.

\bibitem{Dum05}
Adrian Dumitrescu.
\newblock On some monotone path problems in line arrangements.
\newblock {\em Comput. Geom.}, 32(1):13--25, 2005.
\newblock \href {https://doi.org/10.1016/J.COMGEO.2005.01.001}
  {\path{doi:10.1016/J.COMGEO.2005.01.001}}.

\bibitem{Mat91}
Jir{\'{\i}} Matousek.
\newblock Lower bounds on the length of monotone paths in arrangement.
\newblock {\em Discret. Comput. Geom.}, 6:129--134, 1991.
\newblock \href {https://doi.org/10.1007/BF02574679}
  {\path{doi:10.1007/BF02574679}}.

\bibitem{MonienS88}
B.~Monien and I.H. Sudborough.
\newblock Min cut is {NP}-complete for edge weighted trees.
\newblock {\em Theoretical Computer Science}, 58(1–3):209--229, 1988.
\newblock \href
  {https://doi.org/http://dx.doi.org/10.1016/0304-3975(88)90028-X}
  {\path{doi:http://dx.doi.org/10.1016/0304-3975(88)90028-X}}.

\bibitem{Ophelders2022}
Tim Ophelders and Salman Parsa.
\newblock Minimum height drawings of ordered trees in polynomial time: Homotopy
  height of tree duals.
\newblock In Xavier Goaoc and Michael Kerber, editors, {\em 38th International
  Symposium on Computational Geometry (SoCG 2022)}, pages 55:1--55:16. Schloss
  Dagstuhl – Leibniz-Zentrum f\"{u}r Informatik, 2022.
\newblock \href {https://doi.org/10.4230/LIPICS.SOCG.2022.55}
  {\path{doi:10.4230/LIPICS.SOCG.2022.55}}.

\bibitem{Rote25}
G{\"u}nter Rote.
\newblock {NumPSLA} -- an experimental research tool for pseudoline
  arrangements and order types.
\newblock In Jan Kratochv{\'\i{}}l and Giuseppe Liotta, editors, {\em 41st
  European Workshop on Computational Geometry (EuroCG 2025)}, pages 18:1--18:8,
  April 2025.
\newblock URL:
  \url{http://page.mi.fu-berlin.de/rote/Papers/abstract/NumPSLA+--+An+experimental+research+tool+for+pseudoline+arrangements+and+order+types.html},
  \href {https://arxiv.org/abs/2503.02336} {\path{arXiv:2503.02336}}.

\bibitem{WuAPL14}
Yu~Wu, Per Austrin, Toniann Pitassi, and David Liu.
\newblock Inapproximability of treewidth and related problems.
\newblock {\em J. Artif. Intell. Res.}, 49:569--600, 2014.
\newblock \href {https://doi.org/10.1613/jair.4030}
  {\path{doi:10.1613/jair.4030}}.

\end{thebibliography}

\appendix

\section{An instance where the primal-dual sweep uses the maximum
  ropelength}
\label{sec:long-instance}

\setcounter{topnumber}5
\setcounter{totalnumber}5
\renewcommand{\topfraction}{0.8}
\renewcommand{\bottomfraction}{0.8}
\renewcommand{\textfraction}{0.1}

We show another example of how the sweep is performed in the following sequence of figures.   
The construction consists of $n$
pseudolines $c_1,\dots,c_n$,
enumerated in top-to-bottom order at $s$,
that satisfy the following:
\begin{itemize}
    \item For any $i>1$, the first crossing along $c_i$ is with pseudoline $c_1$.
    \item Let $F$ be the face to the left of the last edge of $c_1$.   Then the top chain of $F$ meets all pseudolines except $c_1$.
\end{itemize}
See Figure~\ref{fig:example_scan_1} for the pseudoline arrangement (for $n=7$) and the initial rope and dual rope.

\newcommand{\figurescale}{0.8}

\begin{figure}[ht]
\centering
\includegraphics[scale=\figurescale]{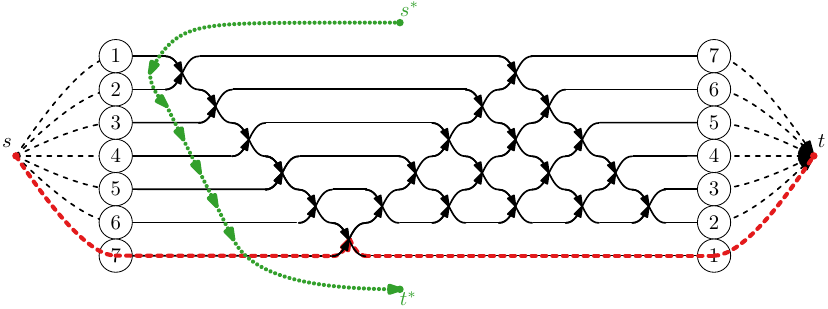}
\vspace{-3mm}
\caption{The arrangement, with initial rope and dual rope.}
\label{fig:example_scan_1}
\end{figure}

We show that if these conditions hold, then the rope-length becomes $2n-2$ at some point
(hence the bound of Claim~\ref{cl:2n-2} is tight).
To see this, observe that the first move is to flip across a face, since the active edge (which is the first edge of pseudoline $c_n$) is bottom incoming.
See Figure~\ref{fig:example_scan_2}.

\begin{figure}[ht]
\centering
\includegraphics[scale=\figurescale,% width=\linewidth,
page=2]{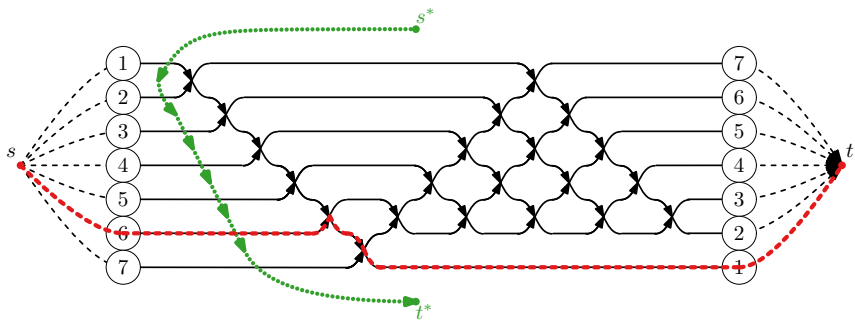}
\vspace{-3mm}
\caption{The situation after the first face-flip.}
\label{fig:example_scan_2}
\end{figure}

The next few moves will \emph{all} be face-flips, because the active edge is always the first edge of pseudoline $c_i$ for some $i>1$, which is bottom incoming because it ends at the intersection with $c_1$.   So we continue face-flips until the active edge is the first edge of $c_1$, and in fact the entire rope is exactly $c_1$.
See Figure~\ref{fig:example_scan_3}.

\begin{figure}[ht]
\centering
\includegraphics[scale=\figurescale,%width=\linewidth,
page=3]{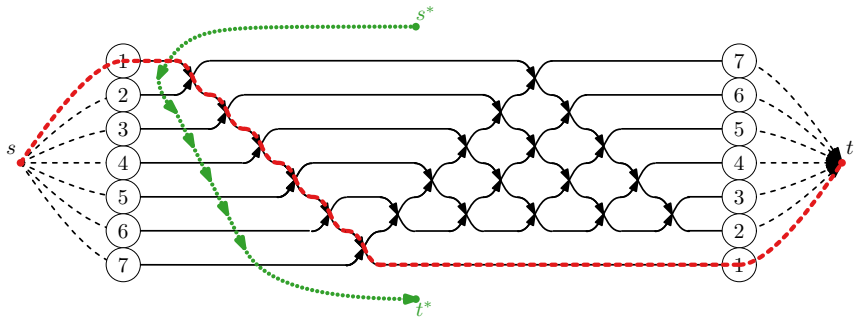}
\vspace{-3mm}
\caption{The situation after repeated face-flips until the rope follows $c_1$.}
\label{fig:example_scan_3}
\end{figure}

Now the active edge is on $c_1$, hence top incoming, and we do a vertex-flip, which pushes the active edge one further down the rope (i.e., along $c_1$).
See Figure~\ref{fig:example_scan_4}.

\begin{figure}[ht]
\centering
\includegraphics[scale=\figurescale,%width=\linewidth,
page=4]{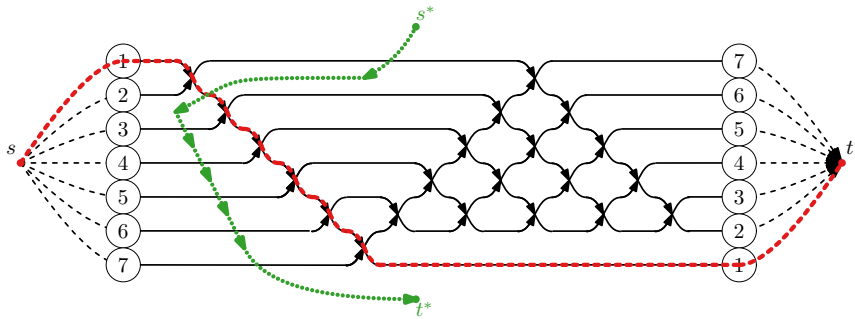}
\vspace{-3mm}
\caption{The situation after the first vertex-flip.}
\label{fig:example_scan_4}
\end{figure}

The next few moves will actually \emph{all} be vertex-flips, because the active edge is always on $c_1$, hence top-incoming if its head is not $t$.   So we continue doing vertex-flips until the active edge is the last edge of $c_1$.
See Figure~\ref{fig:example_scan_5}.

\begin{figure}[htb]
\centering
\includegraphics[scale=\figurescale,%width=\linewidth,
page=5]{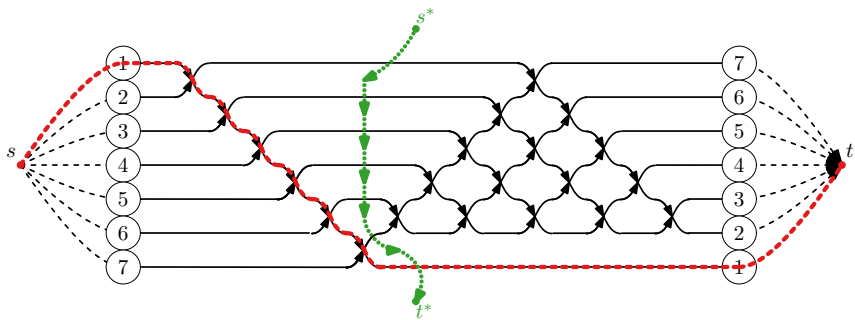}
\vspace{-3mm}
\caption{The situation after repeated vertex-flips until the active edge is the last edge of $c_1$.}
\label{fig:example_scan_5}
\end{figure}

Now the active edge is bottommost incoming
at its head $t$, which means that we do a face-flip at the face $F$ to the left of the active edge.    Recall that we constructed our arrangement so that the upper chain of this face $F$ has length $n-1$.   Also, pseudoline $c_1$ has $n$ edges, of which the rope uses all but the last one.   Therefore at this point the rope has length $2n-2$.
See Figure~\ref{fig:example_scan_6}.

\begin{figure}[htb]
\centering
\includegraphics[scale=\figurescale,%width=\linewidth,
page=6]{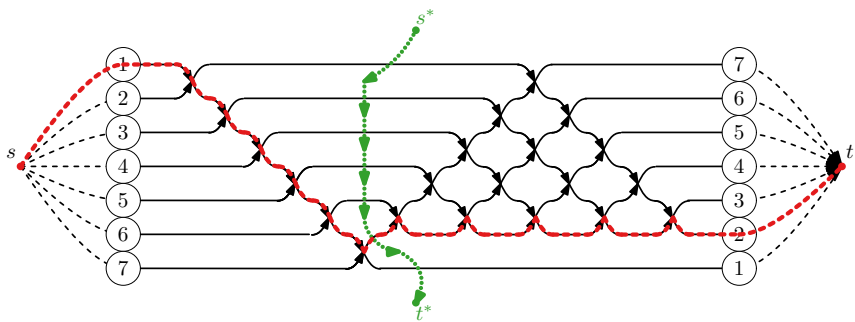}
\vspace{-3mm}
\caption{After one more face-flip, the rope length is $2n-2$.
}
\label{fig:example_scan_6}
\end{figure}

%\therese{If we have too much time at hand we could show a few more moves, but this is very non-urgent.}

We note that rope-length $2n-2$ is not required in this example if we sweep differently.   In particular, a sweep
with rope-length $n+1$ can be obtained % in this example
by applying the algorithm to the reflected arrangement in which left and right are swapped.

\end{document}